\documentclass[prl,twocolumn,showpacs,preprintnumbers,amsmath,amssymb,superscriptaddress]{revtex4}

\usepackage{bm}
\usepackage{graphicx}
\usepackage{amsbsy}
\usepackage{amsmath}
\usepackage{amsfonts}
\usepackage{amsthm}
\usepackage{listings}

\begin{document}

\theoremstyle{plain}
\newtheorem{theorem}{Theorem}
\newtheorem{lemma}[theorem]{Lemma}
\newtheorem{corollary}[theorem]{Corollary}
\newtheorem{conjecture}[theorem]{Conjecture}
\newtheorem{proposition}[theorem]{Proposition}

\theoremstyle{definition}
\newtheorem{definition}{Definition}

\theoremstyle{remark}
\newtheorem*{remark}{Remark}
\newtheorem{example}{Example}

\title{Classification of multipartite entanglement in all dimensions}   

\author{Gilad Gour}\email{gour@ucalgary.ca}
\affiliation{Institute for Quantum Information Science and 
Department of Mathematics and Statistics,
University of Calgary, 2500 University Drive NW,
Calgary, Alberta, Canada T2N 1N4} 
\affiliation{Department of Mathematics, University of California/San Diego, 
        La Jolla, California 92093-0112}
\author{Nolan R. Wallach}\email{nwallach@ucsd.edu}
\affiliation{Department of Mathematics, University of California/San Diego, 
        La Jolla, California 92093-0112}

\date{\today}

\begin{abstract} 
We provide a systematic classification of multiparticle entanglement in terms of equivalence classes of states under stochastic local operations and classical communication (SLOCC). We show that such an SLOCC equivalency class of states is characterized by ratios of homogenous polynomials that are invariant under local action of the special linear group.
We then construct the complete set of all such SL-invariant polynomials (SLIPs). Our construction is based on Schur-Weyl duality and applies to any number of qudits in all (finite) dimensions. In addition, we provide an elegant formula for the dimension of the homogenous SLIPs space of a fixed degree as a function of the number of qudits. The expressions for the SLIPs involve in general many terms, but for the case of qubits we also provide much simpler expressions.
\end{abstract}  

\pacs{03.67.Mn, 03.67.Hk, 03.65.Ud}

\maketitle


Multi-particle entanglement is an essential resource for a variety of quantum information processing tasks. These include conventional~\cite{Nie00} and measurement-based quantum computation~\cite{Rau01}, quantum error correction schemes~\cite{Nie00}, 
quantum secret sharing~\cite{Hil99}, quantum simulations~\cite{Llo96}, and in principle in any task involving entangled many-body quantum systems~\cite{Hor09,Ple07}.  These applications to quantum information, along with its intriguing properties and potential applications to condensed matter physics~\cite{Guh05}, sparked an enormous amount of literature dedicated to the classification of multipartite entanglement.  Nevertheless, our current understanding of multipartite entanglement is very limited, and besides few cases involving small systems, the resource theory of multipartite entanglement is still in its infancy~\cite{Wal12,GW11}. 

Entanglement is a non-local resource with which it is possible to overcome
the limitations imposed by local operations. Therefore, multipartite states are classified according to their inter-convertibility under stochastic local operations and classical communication (SLOCC)~\cite{Vid00}. That is, two states belong to the same entanglement class (or SLOCC equivalent class) if it is possible to reversibly convert one state to the other with non-zero probability using only LOCC.
While in three qubits there are only 6 inequivalent SLOCC classes, in four (or more) qubits there are already uncountable number of inequivalent SLOCC classes~\cite{Vid00,GW10}. This is a simple indication that the complexity in characterizing SLOCC classes grows rapidly as the number of particles increases.  

Initially, much of the literature focused on classification of SLOCC classes in terms of homogeneous polynomial functions of the coefficients of the pure multi-qudit state in question. If these polynomials on $N$-qudit systems are also invariant under $SL(d,\mathbb{C})^N$ (here $SL(d,\mathbb{C})$ is the set of all $d\times d$ complex matrices with determinant 1) then their absolute  values are measures of genuine multipartite entanglement~\cite{Ver03,Vla12}. These entanglement measures have many interesting and desirable properties, including initial-state-independence behavior under local noise~\cite{GG10,Vla12}. For pure two qubits and three qubits systems, the concurrence~\cite{Woo98} and the 3-tangle~\cite{CKW} are the unique polynomials of this kind. For four qubits, however, there are infinitely many such invariants and the set of all of them is generated by four invariants~\cite{Luq03,GW10}. Beyond that, and despite the extensive literature, very little is known about the set of all such SL-invariant polynomials (SLIPs) except for few techniques that were used to construct some of the SLIPs for multi-qubit states~\cite{Lei04,Ost05}. 

In this Letter we show that almost all SLOCC equivalent classes, including the dense set of stable SLOCC classes, can be distinguished by ratios of homogeneous SLIPs of the same degree. We then find an algorithm to construct \emph{all} of the SLIPs. Our main idea is to look at homogenous polynomials of a fixed degree $k$, and view them as vectors in a Hilbert space consisting of $k$-copies of the original Hilbert space. This identification enables us to use the Schur-Weyl duality and other techniques from representation and invariant theory to construct the full set of homogeneous SLIPs of degree $k$. Our technique can be applied to any number of qudits in all dimensions. For the case of qubits, we also find an elegant way to express many of the SLIPs.

We consider here the Hilbert space of $n$ qudits which we denote by:
$$
\mathcal{H}_n\equiv \mathbb{C}^{m_1}\otimes\mathbb{C}^{m_2}\otimes\cdots\otimes\mathbb{C}^{m_n}\;.
$$
Local (generalized) measurements on each of the $n$ qudits, can transform probabilistically an initial pure state $|\psi\rangle\in\mathcal{H}_n$  to some other pure state  $|\phi\rangle\equiv A_1\otimes\cdots\otimes A_n|\psi\rangle$, where $\{A_j\}_{j=1}^{n}$ are $m_j\times m_j$ matrices. If this transformation is reversible then all the matrices $\{A_j\}$ can be taken to be invertible~\cite{Vid00}, and $\psi$ and $\phi$ are said to be SLOCC equivalent. We are interested here in characterizing classes of SLOCC equivalent states.

Up to normalization, an invertible SLOCC map $A_1\otimes\cdots\otimes A_n$ can be described as an element of the group
$$
 G\equiv \text{SL}(m_1,\mathbb{C})\otimes\text{SL}(m_2,\mathbb{C})\otimes\cdots\otimes\text{SL}(m_n,\mathbb{C})
$$
where $SL(m_j,\mathbb{C})$ is the group consisting of all $m_j\times m_j$ complex invertible matrices with determinant 1. This last condition implies that the states $g|\psi\rangle$ for $g\in G$ are in general not normalized even if $|\psi\rangle$ is normalized. The orbit $G|\psi\rangle:=\{g|\psi\rangle\;|\;g\in G\}$ is therefore consisting of non-normalized states in the SLOCC equivalent class of $|\psi\rangle$. Working with non-normalized states enables us to characterize the distinct orbits $\{G|\psi\rangle\}$ with SLIPs, while the distinct SLOCC classes can be characterized by ratios of the SLIPs. 

A  polynomial $f:\mathcal{H}_{n}\to\mathbb{C}$ that satisfies
$$
f(g|\psi\rangle)=f(|\psi\rangle)\;\;\;\forall\;g\in G\;\;\text{and}\;\;\forall\;|\psi\rangle\in\mathcal{H}_{n}\;,
$$
is called $SL$-invariant polynomial (SLIP). The set of all SLIPs form a vector space over $\mathbb{C}$. 
We can always choose the basis of this polynomial space 
to consist of \emph{homogeneous} SLIPs~\footnote{Any polynomial can be written as
a linear combination of homogeneous polynomials, and the invariance follows from the fact that $G$ acting on $\mathcal{H}_n$ linearly.}.
For this reason, we will focus on homogeneous SLIPs of some degree $k\in\mathbb{N}$. 
The dimension of the space of homogeneous polynomials of fixed degree $k$ is finite, although we will see that it grows exponentially with $n$. 
In the supplemental material we show that the degree $k$ must be divisible by the least common 
multiple of the integers $m_1,m_2,...,m_n$. This implies for example that for qubits there are only SLIPs of even degree.

We show now that SLIPs can be used to determine whether two states in 
$\mathcal{H}_n$ belong to the same SLOCC class. Indeed, suppose that 
$|\psi\rangle,|\phi\rangle\in\mathcal{H}_n$ are two (normalized) states belonging to the same SLOCC class.
That is, there exists a phase $\theta\in[0,2\pi)$ and $g\in G$ such that
\begin{equation}\label{orbit}
|\psi\rangle=e^{i\theta}\frac{g|\phi\rangle}{\|g|\phi\rangle\|}\;.
\end{equation}
Suppose $f_k$ is a homogenous SLIP of degree $k$. Then,
\begin{equation}\label{psiphi}
f_k(|\psi\rangle)=\frac{e^{i\theta k}}{\|g|\phi\rangle\|^k}f_k(|\phi\rangle).
\end{equation}
Suppose now that $h_k$ is another homogeneous SLIP of the \emph{same} degree $k$. Since it also satisfies the relation~\eqref{psiphi}, we get
\begin{equation}\label{ratio}
\frac{f_k(|\psi\rangle)}{h_k(|\psi\rangle)}=\frac{f_k(|\phi\rangle)}{h_k(|\phi\rangle)}\;,
\end{equation}
assuming $h_k(|\psi\rangle)\neq 0$. Thus, if $|\psi\rangle$ and $|\phi\rangle$
belong to the same SLOCC class they must satisfy~\eqref{ratio} for any pair of two
homogeneous SLIPs of the same degree (this direction was pointed out recently in~\cite{Oli11}). Here we show that the converse also holds for almost all states in $\mathcal{H}_n$ (see details of the proof in the supplementary material): 
\begin{proposition}
Let $|\psi\rangle,|\phi\rangle\in\mathcal{H}_n$ be two stable states (i.e. states $|\psi\rangle$ whose orbits $G|\psi\rangle$ are closed).
Then, there exists $\theta\in[0,2\pi)$ and $g\in G$ such that~\eqref{orbit} holds if and only if~\eqref{ratio} holds for all homogeneous SLIPs of degree $k$ with $h_k(|\psi\rangle)\neq 0$.
\end{proposition}

The set of all stable states is open and dense in $\mathcal{H}_n$. It contains almost all states in $\mathcal{H}_n$ including many non-generic states. In particular, an orbit $G|\phi\rangle$ is stable if and only if it contains 
a state with the property that the reduced density matrix 
on each single qudit (after tracing out the other $n-1$ qudits) is proportional to 
the identity~\cite{GW11,GW10}. Thus, many interesting and physically relevant states are stable. This includes all quantum error-correcting code states, all GHZ states, all cluster states, and in general all graph states. Hence, our construction below of all SLIPs provides an almost-complete classification of $\mathcal{H}_n$ into SLOCC equivalent classes.

SLIPs involve in general cumbersome expressions. However, for the case of $n$ qubits there exists
a large class of polynomials that can be expressed very elegantly. The main reason for that is the existence of the matrix
\begin{equation}\label{J}
J\equiv-i\sigma_y=\left(\begin{array}{cc}0&1\\-1&0\end{array}\right)
\end{equation}
for which any $2\times 2$ matrix $A\in SL(2,\mathbb{C})$ satisfies
$$
A^TJA=J\;,
$$
where $A^T$ denotes the transposed matrix of $A$. In higher dimensions, such an invariant matrix $J$ that satisfies the above equation for all
$A\in SL(d,\mathbb{C})$ does not exist.
Thus, before introducing our main results, 
we first show that this property of $J$ leads to a large class of SLIPs in the qubit case where
$m_j=2$ for all $j=1,2,...,n$.

The above property of $J$ implies that for any two states 
$|\psi\rangle,|\phi\rangle\in\mathcal{H}_n$, and any $g\in G$,
$$
(g\psi,g\phi)_n=(\psi,\phi)_n\;,
$$
where the bilinear form $(\cdot,\cdot)_n$ is defined by
\begin{equation}\label{blin}
(\psi,\phi)_n\equiv\langle\psi^*|J\otimes\cdots\otimes J|\phi\rangle\;.
\end{equation}
Here $J$ appears $n$-times, and the vector $|\psi^*\rangle$ is the complex conjugate
of $|\psi\rangle$ when it is expressed in a fixed basis such that each $J$ has the form~(\ref{J}).
Therefore, $(\psi,\psi)_n$ is an homogenous SLIP of degree 2. Note however that for odd $n$, $(\psi,\psi)_n=0$
since in that case $J^{\otimes n}$ is a skew-symmetric matrix~\footnote{This is consistent with the fact that for odd number of qubits the lowest degree of a (non-zero) homogenous SLIP is 4}. 

For any choice of $q$ qubits in $\mathcal{H}_n$ with $1\leq q<n$, we associate a bipartite cut $\mathcal{A}_q\otimes\mathcal{B}_{n-q}$ between the chosen $q$ qubits,$A_q$, and the remaining $n-q$ qubits, $B_{n-q}$.
For a given $q$ there are ${n\choose q}$
bipartite cuts of the form $\mathcal{H}_n\cong\mathcal{A}_q\otimes\mathcal{B}_{n-q}$. 
Now,
fix a bipartite cut $\mathcal{A}_q\otimes\mathcal{B}_{n-q}$. Any state $|\psi\rangle\in \mathcal{A}_q\otimes\mathcal{B}_{n-q}$
can be written as 
\begin{equation}\label{psi}
|\psi\rangle=\sum_{j=1}^{2^q}\sum_{k=1}^{2^{n-q}}a_{jk}|u_j\rangle|v_k\rangle\;,
\end{equation}
where $\{|u_j\rangle\}$ and $\{|v_k\rangle\}$ are orthonormal bases of $\mathcal{A}_q$ 
and $\mathcal{B}_{n-q}$, respectively. For any such $|\psi\rangle$, we denote by $A$ the
$2^q\times 2^{n-q}$ matrix whose elements are $a_{jk}$. 
Furthermore, we denote by $U$
the $2^q\times 2^q$ matrix whose elements are $U_{jj'}=(u_j,u_{j'})_q$, and by $V$ the
$2^{n-q}\times 2^{n-q}$ matrix whose elements are $V_{kk'}=(v_k,v_{k'})_{n-q}$,
where $({\bf \cdot},{\bf\cdot})$ is the bilinear form defined in Eq.(\ref{blin}). With this notation we get (see more details in the supplementary material) that for any $\ell\in\mathbb{N}$ the function
\begin{equation}
\label{mn}
 f_{\ell}^{\mathcal{A}_q |\mathcal{B}_{n-q}}(\psi)\equiv{\rm Tr}\left[\left(UAVA^T\right)^{\ell}\right]
 \end{equation}
 is an SLIP of degree $2\ell$.
The function $ f_{\ell}^{\mathcal{A}_q |\mathcal{B}_{n-q}}(\psi)$ is invariant under a change of bases in $\mathcal{A}_q$ 
and $\mathcal{B}_{n-q}$. Moreover, since the form $(\cdot,\cdot)_q$ is symmetric for even $q$, if $q$ is even we can choose the basis $\{|u_j\rangle\}$ such that the matrix $U=I$.
Similarly, if $n-q$ is even then we can choose the basis $\{|v_k\rangle\}$ such that the matrix $V$ is the identity.

The SLIPs defined above are not all independent. 
Different bipartite cuts can lead to the same polynomial. 
For example, in the case of even number of qubits, all the polynomials $f_{\ell=1}^{\mathcal{A}_q |\mathcal{B}_{n-q}}(\psi)$,
corresponding to ${n\choose q}$ 
bipartite cuts $\mathcal{A}_q\otimes\mathcal{B}_{n-q}$, are the same (up to a constant factor) since
there is only one SLIP of degree 2 given by $(\psi,\psi)_n$~\cite{W}. 
It is therefore natural to ask if all SLIPs can be written as linear combinations of $\{f_{\ell}^{\mathcal{A}_q |\mathcal{B}_{n-q}}\}$. It is true for degrees 2 and 4, but already in 5 qubits there exists (see below) a unique polynomial of degree 6, whereas our polynomials $f_{\ell}^{\mathcal{A}_q |\mathcal{B}_{n-q}}$ can not be of degree 6 if $n$ is odd. We therefore need a more systematic way to construct all SLIPs.

We are now ready to introduce our main technique to
compute the $G$-invariants in the space of polynomials on $\mathcal{H}_n$ of
degree $k$. We denote the set of all these SLIPs by $I_{k,\mathbf{m}}$
where $\mathbf{m}=(m_{1},...,m_{n})$ is the vector of the dimensions of the $n$-qudits. Observe that any homogeneous polynomial on $\mathcal{H}_n$ of degree $k$ can be written as an inner product between some vector $|v\rangle\in \otimes^{k}\mathcal{H}_n$ and $k$-copies of a vector $|x\rangle\in\mathcal{H}_n$:
\begin{equation}\label{fv}
f_{v}(x)=\left\langle v|\otimes^{k}x\right\rangle\;.
\end{equation}
Hence, any homogeneous polynomial on $\mathcal{H}_n$ of degree $k$  corresponds to some 
(not necessarily unique) vector $|v\rangle\in\otimes^k\mathcal{H}_n$. Note that if we also have $\left(\otimes^k g\right)|v\rangle=|v\rangle$ for all $g\in G$, then the polynomial $f_v(x)$ is $G$-invariant. Denoting by 
$\left(\otimes^{k}\mathcal{H}_n\right)  ^{G}$ all such $G$-fixed vectors of
$\otimes^{k}\mathcal{H}_n$, we conclude: 
$$
\left\{f_{v}\;\big|\;v\in\left(  \otimes^{k}H\right)  ^{G}\right\}=I_{k,\mathbf{m}}.
$$
Note however that in general distinct vectors in $\otimes^{k}\mathcal{H}_n$ can lead to the \emph{same} polynomial and, in particular, we will see below that
$\dim\left(  \otimes^{k}H\right)  ^{G}$ is much larger than $\dim
I_{k,\mathbf{m}}$. 

We now show how this view of SLIPs leads to an algorithm for finding all of them. We first write out $v$ in terms of the standard basis in each tensor slot. That is,%
\[
v=\sum a_{i_{1}i_{2}\cdots i_{k}}\left\vert i_{1}i_{2}\cdots i_{k}%
\right\rangle
\]
and each $|i_{j}\rangle\in\mathcal{H}_n$ is an element in the computational basis of $\mathcal{H}_n$, i.e. it is of the form
$$
|i_j\rangle=\left\vert b_{1j}b_{2j}\cdots b_{nj}\right\rangle 
$$ with $1\leq b_{sj}\leq m_{s}$ for $s=1,2,...,n$.
Writing
\[
|x\rangle=\sum x_{i}\left\vert i\right\rangle
\]
relative to the same computational basis of $\mathcal{H}_n$ we have%
\[
f_{v}(x)=\sum\overline{a_{i_{1}i_{2}\cdots i_{k}}}x_{i_{1}}x_{i_{2}}\cdots
x_{i_{k}}.
\]
This shows that if we symmetrize the elements relative to the action of
$S_{k}$ (permuting the factors) then the symmetric elements yield
$I_{k,\mathbf{m}}$ unambiguously. In what follows, we will view 
$\otimes^{k}\mathcal{H}_n$ as a representation space of $G$ with the homomorphism $g\to\otimes^k g$ for any $g\in G$. 

\begin{figure}[tp]
\includegraphics[scale=.40]{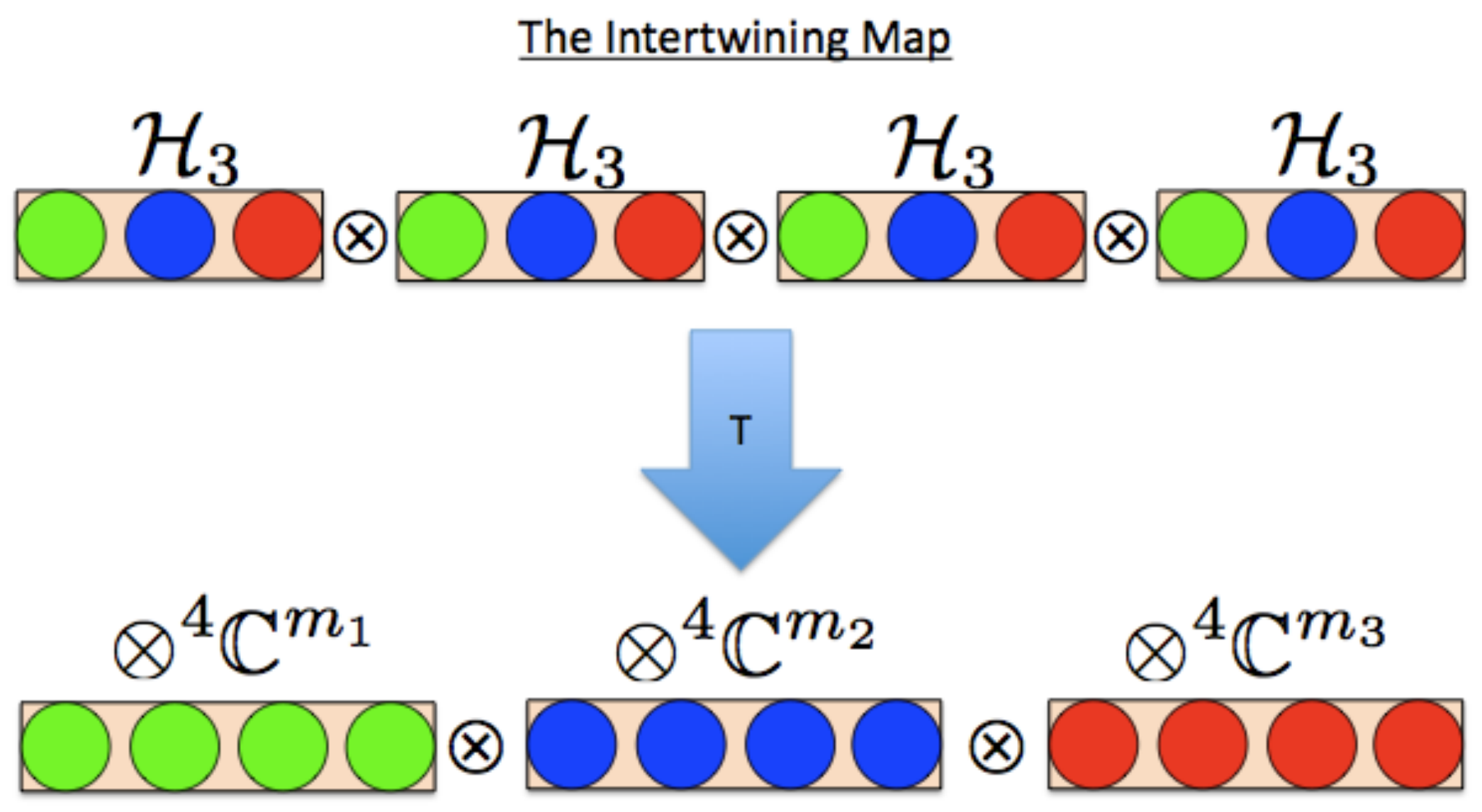}
\caption{The intertwining operator, $T$, for $n=3$ and $k=4$. The Hilbert space of the three qudits is denoted by
$\mathcal{H}_{3}=\mathbb{C}^{m_1}\otimes\mathbb{C}^{m_2}\otimes\mathbb{C}^{m_3}$, where the three colours correspond to the three dimensions $m_1$, $m_2$, and $m_3$.}
\end{figure}

Our task is to find $\left(  \otimes^{k}\mathcal{H}_n\right)  ^{G}$. 
The key observation is that the vector space $\otimes^{k}\mathcal{H}_n$
is isomorphic to the vector space
\[
\left(  \otimes^{k}\mathbb{C}^{m_{1}}\right)  \otimes\left(  \otimes
^{k}\mathbb{C}^{m_{2}}\right)  \otimes\cdots\otimes\left(  \otimes
^{k}\mathbb{C}^{m_{n}}\right)\;.
\]
The isomorphism map
(which is a `transpose'-type map) 
maps the computational basis element $\left\vert i_{1}
i_{2}\cdots i_{k}\right\rangle$ to the basis element $\left\vert i^{1}i^{2}...i^{n}%
\right\rangle $, where
\begin{align}
& |i_j\rangle\equiv\left\vert b_{1j}b_{2j}\cdots b_{nj}\right\rangle\in\mathcal{H}_n\nonumber\\
& |i^{s}\rangle\equiv\left\vert b_{s1}b_{s2}\cdots b_{sm_s}\right\rangle \in\otimes
^{k}\mathbb{C}^{m_{s}},
\end{align}
and for each $1\leq s\leq n$, $\{|b_{sj}\rangle\}_{j=1}^{m_s}$ is the standard basis
of $\mathbb{C}^{m_{s}}$.
This map is an intertwining operator (see Fig. 1). We will denote it by
\[
|v\rangle\longmapsto |v^{T}\rangle.
\]

Quite remarkably, under this intertwining map, the problem of finding
$\left(  \otimes^{k}\mathcal{H}_n\right)  ^{G}$ is reduced to finding
$\left(  \otimes^{k}\mathbb{C}^{m_s}\right)  ^{SL(m_s,\mathbb{C)}}$ for each $s=1,2,...,n$. To see it, we use the same notation for the inverse map from $\left(  \otimes
^{k}\mathbb{C}^{m_{1}}\right)  \otimes\left(  \otimes^{k}\mathbb{C}^{m_{2}
}\right)  \otimes\cdots\otimes\left(  \otimes^{k}\mathbb{C}^{m_{n}}\right)  $
to $\otimes^{k}\mathcal{H}_n$. Let $P_{m,k}$ denote the orthogonal
projection of $\otimes^{k}\mathbb{C}^{m}$ onto $\left(  \otimes^{k}
\mathbb{C}^{m}\right)  ^{SL(m,\mathbb{C)}}$ (which is $0$ unless $k$ is
divisible by $m).$ The map
\begin{equation}\label{m1}
P(v)=\left(  P_{m_{1},k}\otimes P_{m_{2},k}\otimes\cdots\otimes P_{m_{n}
,k}(v^{T})\right)  ^{T}
\end{equation}
is the orthogonal projection from $\otimes^{k}\mathcal{H}_n$ onto $\left(  \otimes
^{k}\mathcal{H}_n\right)  ^{G}$. This is a
simple consequence of the fact that a state being invariant under
$g^{\otimes k}$ has to be in the span of $\{|v_1\rangle,...,|v_n\rangle\}$, where
$|v_i\rangle$ are elements of $\left(  \otimes^{k}\mathbb{C}^{m_i}\right)  ^{SL(m_i,\mathbb{C)}}$.

It is therefore left to compute $P_{m,k}$. This is a classic problem in representation theory and can be solved using the Schur-Weyl duality. The Schur-Weyl duality
relates the irreducible representations (irreps) of $SL(m,\mathbb{C})$ 
(when acting on $\otimes^{k}\mathbb{C}^{m}$) 
and the symmetric group on $k$ elements, $S_k$ 
(with the natural action on $\otimes^{k}\mathbb{C}^{m}$).
Recall that $k$ must be of the form $k=mr$ for some $r\in\mathbb{N}$ in order
to have $P_{m,k}\neq0.$ Let $\chi_{\lambda}$ be the character of $S_{k}$
corresponding to the partition of $k$ given by $\lambda=(r,r,...,r)$ ($m$ $r$'s). Then (see more details in the supplementary material),
\begin{equation}\label{m2}
P_{m,k}=\frac{d_{\lambda}}{k!}\sum_{\sigma\in S_{k}}\chi_{\lambda}
(\sigma)\sigma\;,
\end{equation}
where $d_\lambda$ is the dimension of the irrep corresponding to the partition $\lambda$.
Equations~\eqref{m1} and~\eqref{m2} determine the projection of $\otimes^{k}\mathcal{H}_n$
onto $\left(  \otimes^{k}\mathcal{H}_n\right)  ^{G}$, and thereby via Eq.~\eqref{fv} provide the complete set of \emph{all} SLIPs of degree $k$. 

We now give a simple example, illustrating how to use Eq.(\ref{m2}) to construct SLIPs.
More examples with more details can be found in the supplementary material.
Consider a system consisting of $n$ qudits each of dimension $m=4$.
The smallest degree possible in this case is $k=4$ for which $r=1$. Thus, in this case $P_{m,k}$ is a projection on $\mathbb{C}^4\otimes\mathbb{C}^4\otimes\mathbb{C}^4\otimes\mathbb{C}^4$. From Eq.~\eqref{m2} we can define $P_{m,k}$ in terms of its action on basis elements as
\[
v_{1}\otimes v_{2}\otimes v_{3}\otimes v_{4}\longmapsto\frac{1}{24}%
\sum_{\sigma\in S_{4}}sgn(\sigma)v_{\sigma1}\otimes v_{\sigma2}\otimes
v_{\sigma3}\otimes v_{\sigma4}.
\]
Substituting this $P_{m,k}$ into Eq.(\ref{m1}) gives the degree $4$ invariant
for an arbitrary number of $n$ qudits each of dimension four. 

The number of terms in the sum of~\eqref{m2} is $k!$, which already for $k=6$ gives $720$ terms. Therefore, using fundamental concepts from invariant theory, we find (see supplementary material) an alternative way to express all the SLIPs of $n$ qubits. This alternative way (which only applies for the qubit case) is slightly more computationally efficient when $k$ is small. As an example we give an explicit formula for the unique SLIP of degree 6 in five qubits (which can not be written in the form~\eqref{mn}).

Our techniques also enable us to write down a simple algorithm to compute the dimension, $d(k,n)$, of the space of SLIPs of degree $k$ in $n$ qubits. We get for
example:
\begin{align*}
& d(2,n)=\frac{1}{2}(1+(-1)^n)\;\;\;\;,\;\;\;\;
 d(4,n)=\frac{2^{n-1}+(-1)^n}{3}\;,\\
& d(6,n)=\frac{1}{144}\left(36+44\cdot(-1)^n+8\cdot2^n+3\cdot(-3)^n+5^{n-1}\right)
\end{align*}
In supplementary material we attach a program in mathematica that computes
$d(k,n)$ for large $n$ and $k$. Moreover, in the more general case of $n$ \emph{qudits}, with $m_1=\cdots=m_n\equiv m$,
the dimension $d(k,n)$ of the space of SLIPs of degree $k=mr$ satisfies (for $k>4$)
$$
\lim_{n\to\infty}\frac{d(k,n)}{(C_{m-1,r})^n}=\frac{1}{k!},\;C_{m-1,r}\equiv (mr)!\frac{\prod_{1 \leq i<j \leq m} (j-i)}{\prod_{j=0}^{m-1}(r+j)!}
$$
where $C_{m,r}$ is the generalization of the $r^\text{th}$ Catalan number (see the supplementary material for more details). 
For $k=4$ (which is only possible in the qubit $m=2$ case) the limit above is 1/6.
The distinction between $k=4$ and $k>4$ follows from the fact that the alternating group on $n$ letters is simple (i.e. without non-trivial normal subgroups) if and only if $k>4$. 
This formula indicates that the dimension of the space of SLIPs grows exponentially both in $n$ and in $k$. 

In conclusion, we showed that SLIPs on $n$ qudits can be used to 
classify almost all SLOCC classes. We constructed a large class of elegant SLIPs for the case of $n$ qubits and then introduced a general technique to construct all SILPs in all dimensions.
Instead of fixing the number of qudits and then trying to calculate all SLIPs in any degree, we fixed the degree $k$ and then where able to calculate all SLIPs of that degree in any number of qudits. This approach also enabled us to calculate the dimensions of the space of homogeneous SLIPs of a given degree.

\emph{Acknowledgments:---}
The authors are grateful for Markus Grassl for many fruitful comments on the first version of the paper.
G.G. research is partially supported by
NSERC and by the Air Force Office of Scientific Research
as part of the Transformational Computing in Aerospace
Science and Engineering Initiative under grant FA9550-12-1-0046.
N.W. research is partially supported by an NSF summer grant.

\begin{titlepage}
\center{\large\textbf{Supplementary Information}\\}
\center{~{ }\\}

\end{titlepage}

\onecolumngrid

\appendix

\section{(i) The degrees of SLIPs}

Let the notion be as in the main text. Recall that we use the acronym SLIPs
for SL invariant polynomials In particular, the set of all homogeneous SLIPs
of degrees $k$ is denoted by $I_{k,\mathbf{m}}$ where $\mathbf{m}%
=(m_{1},...,m_{n})$ is the vector of the dimensions of each of the $n$%
-qudits.

\begin{proposition}
Let $r=lcm(m_{1},...,m_{n})$. Then $I_{k,\mathbf{m}}=0$ unless $k=qr$
with $q$ a non-negative integer.
\end{proposition}

\begin{proof}
Recall that the center of $SL(m,\mathbb{C)}$ is the set of scalar matrices $%
\{\zeta I|\zeta^{m}=1\}$. Therefore, if $f_{k}\in I_{k,\mathbf{m}}$ then 
\begin{equation*}
f_{k}(|\psi\rangle)=f_{k}(\zeta|\psi\rangle)=\zeta^{k}f_{k}(|\psi\rangle), 
\end{equation*}
for any $|\psi\rangle\in\mathcal{H}_{n}$ and any $\zeta\in\mathbb{C}$ such
that $\zeta^{m_{j}}=1$ for some $1\leq j\leq n$. The first equality follows
from the $G$-invariance of $f_{k}$ and the second equality follows from its
homogeneity. Thus, each $m_{j},1\leq j\leq n,$ must divide $k$.
\end{proof}

\subsection{(ii) SLOCC classes can be distinguished by SLIPs}

\begin{definition}
A state $|\psi\rangle\in\mathcal{H}_n$ is said to be \emph{stable} if its orbit $G|\psi\rangle$ is closed.
\end{definition}
The next proposition will explain this terminology.
Almost all states in~$\mathcal{H}_n$ are stable (see for example~\cite{W}). 
Moreover, stable orbits have the following characterization~\cite{GW11,GW10}:
Let $|\psi\rangle\in\mathcal{H}_n$. Then, $|\psi\rangle$ is stable if and only if there exists a (non-normalized) state
$|\phi\rangle\in G|\psi\rangle$ with the property that all its $n$ local reduced density matrices are proportional to the identity matrix. Therefore, for example, all states in a code space of an error correcting code that is capable of correcting one erasure error must be stable.

The set of states that are not stable is also interesting and is of measure zero in $\mathcal{H}_n$. It contains the null cone.
The null cone is the set of states on which all SLIPs vanish. Therefore, states and orbits in the null cone can not be characterized with SLIPs. However, there are also states that are semi-stable in the sense that they are not stable, but they
are also not in the null cone. The following lemma provides a characterization for such semi-stable states.

\begin{proposition}\label{lem2}
Let $|\psi\rangle\in\mathcal{H}$ be a non-stable state, and suppose $|\psi\rangle$ is not in the null cone. 
Then, there exists a stable state $|\eta\rangle$ and a state in the null cone $|\zeta\rangle$ such that
$$
|\psi\rangle=|\eta\rangle+|\zeta\rangle\;.
$$
Moreover, denote by $G_\eta\equiv\left\{h\in G\;\big|\;h|\eta\rangle=|\eta\rangle\right\}$ 
the stabilizer group of $|\eta\rangle$. Then, there exists a group homomorphism 
$\varphi\;:\;\{z\in\mathbb{C}\;\big|\;z\neq 0\}\to G_\eta$ such that
$$
\lim_{z\to 0}\varphi(z)|\zeta\rangle=0\;.
$$
\end{proposition}

\begin{proof}
This is a consequence of the (generalized) Hilbert-Mumford theorem (c.f.~\cite{GIT}). The theorem asserts that there exists
$\phi: \mathbb{C}^\times \rightarrow G$ an algebraic group homomorphism such that the limit of
$\phi(z)\left\vert \psi \right\rangle$ as $z \rightarrow 0$ is a stable state  $\left\vert\eta\right\rangle$. It is
obvious that $\phi(z)\left\vert \eta \right\rangle = \left\vert\eta\right\rangle$ for all $z$.
The theorem follows from these assertions.
\end{proof}
We are now ready to introduce the main proposition of this subsection.
\begin{proposition}\label{prop2}
Let $|\psi_1\rangle,|\psi_2\rangle\in\mathcal{H}_{n}$ be two normalized stable states in $\mathcal{H}_n$. Then, for
all $k\in\mathbb{Z}_{>0}$, and for all $f_{k},h_{k}\in I_{k,\mathbf{m}}$
such that $h_{k}(|\psi_1\rangle)\neq0$, 
\begin{equation*}
\frac{f_{k}(|\psi_1\rangle)}{h_{k}(|\psi_1\rangle)}=\frac{f_{k}(|\psi_2\rangle )}{%
h_{k}(|\psi_2\rangle)}
\end{equation*}
if and only if there exists $g\in G$ and a phase $\theta\in\lbrack0,2\pi]$
such that 
\begin{equation}\label{slo}
|\psi_1\rangle=e^{i\theta}\frac{g|\psi_2\rangle}{\left\Vert g|\psi_2\rangle
\right\Vert }. 
\end{equation}
Moreover, if $|\psi_1\rangle$ and $|\psi_2\rangle$ are not stable then the proposition still holds if $|\psi_1\rangle$ and $|\psi_2\rangle$ in Eq.(\ref{slo}) are replaced with $|\eta_1\rangle$ and $|\eta_2\rangle$, respectively. Here for $j=1,2$
$|\psi_j\rangle=|\eta_j\rangle+|\zeta_j\rangle$ as in Proposition~\ref{lem2}. 
\end{proposition}

\begin{remark}
If $n>1$ and all of the $m_{i}$ are equal then the set of states with closed
orbits (under $G$) contains an open dense set of $\mathcal{H}_{n}$. Every
orbit contains a unique closed orbit in its closure~\cite{Bo}. For a reductive
algebraic group two closed orbits are equal if and only if all of the values of the invariant ploynomials
agree on them (we will use this fact in the proof of the proposition c.f.~\cite{GIT}).
\end{remark}

\begin{proof}
Note that if there exist $k\in\mathbb{Z}_{>0}$ and $h_{k}\in I_{k,\mathbf{m}}
$ such that $h_{k}(|\psi_1\rangle)\neq0$ but $h_{k}(|\psi_2\rangle)=0$, then
clearly $|\psi_1\rangle$ and $|\psi_2\rangle$ are not SLOCC equivalent. We
therefore assume without loss of generality that there exists $h_{k}\in I_{k,%
\mathbf{m}}$ such that $h_{k}(|\psi_1\rangle)\neq0$ and $h_{k}(|\psi_2\rangle)%
\neq0$. Fixing $h_{k}$, and writing $\lambda=h_{k}(|\psi_1\rangle)/h_{k}(|\psi_2%
\rangle)\neq0$, our hypothesis implies that $f_{k}(|\psi_1\rangle)=\lambda
f_{k}(|\psi_2\rangle)$ for all $f_{k}\in I_{k,\mathbf{m}}$. Now, for an
arbitrary degree $l\in\mathbb{N}$, let $f_{l}\in I_{l,\mathbf{m}}$ be some
homogeneous SLIP of degree $l$. Then, $(f_{l})^{k}$ is an SLIP of the same
degree as $(h_{k})^{l}$. Hence, the hypothesis implies that

\begin{equation*}
f_{l}(|\psi_1\rangle^{k}=\lambda^{l}\left( f_{l}(|\psi_2\rangle)\right)
^{k}=f_{l}(\lambda^{1/k}|\psi_2\rangle)^{k}. 
\end{equation*}
So if $f_{l}\in I_{l,\mathbf{m}}$ then

\begin{equation*}
f_{l}(|\psi_1\rangle)=\zeta_{l}f_{l}\left( \lambda^{1/k}|\psi_2\rangle\right) , 
\end{equation*}
with $\zeta_{l}^{k}=1$. At first glance it looks like the $k^{\text{th}}$
root of unity $\zeta_{l}$ depends on the polynomial $f_{l}$. We will see
that it is very constrained. We first show that it depends only on the
degree $l$. Indeed, think of $\delta_{\psi_1}(f_{l})=f_{l}(|\psi_1\rangle)$ as a
linear functional on $I_{l,\mathbf{m}}$. Then $\ker\delta_{\psi_1}=\ker\delta
_{\lambda^{\frac{1}{k}}\psi_2}$ which implies that $\delta_{\psi_1}$ is
proportional to $\delta_{\lambda^{\frac{1}{k}}\psi_2}$. Thus, $\zeta_{l}$
depends only on the degree $l$. We now consider the group $\widetilde{G}%
=\mu_{k}G$ where $\mu_{k}=\{\zeta I|\zeta^{k}=1\}$. Note that $\widetilde{G}$
is also a reductive algebraic subgroup of $GL(\mathcal{H}_{n})$ so its
closed orbits are separated by its invariant polynomials. Furthermore every
vector that has a closed orbit under $G$ has a closed orbit under $%
\widetilde{G}$. Let $f_{1},...,f_{m}$ be a set of homogeneous generators for
the $G$-invariant polynomials on $\mathcal{H}_{n}$ with degrees $%
d_{1},...,d_{m}$. Then the $\widetilde{G}$ invariant polynomials on $%
\mathcal{H}_{n}$ are the linear combinations $%
\sum_{i_{1},...,i_{m}}a_{i_{1},...,i_{m}}f_{1}^{i_{1}}\cdots f_{m}^{i_{m}}$
where the sum \ is taken over indices with $i_{1}d_{1}+...+i_{m}d_{m}$
divisible by $k$. Repeating the same arguments for the group $\widetilde{G}$
(note that $h_{k}$ is $\widetilde{G}$--invariant) we conclude that if $p$ is
a $\widetilde{G}$ invariant polynomial on $\mathcal{H}_{n}$ then $%
p(|\psi_1\rangle)=p(\lambda^{\frac{1}{k}}|\psi_2\rangle)$. Thus, since the
orbits are closed, we find that $|\psi_1\rangle=\zeta \lambda^{\frac{1}{k}%
}g|\psi_2\rangle$ with $\zeta^{k}=1$ and $g\in G$. The upshot is $%
|\psi_1\rangle=cg|\psi_2\rangle$ with $c\in\mathbb{C},c\neq0,g\in G$. But $%
\left\Vert |\psi_1\rangle\right\Vert =1$ which implies that $|c|\left\Vert
g|\psi_2\rangle\right\Vert =1$. Thus $c=\frac{e^{i\theta}}{\left\Vert
g|\psi_2\rangle\right\Vert }$ as was to be proved.
\end{proof}

\section{(iii) A class of simple SLIPs for $n$ qubits}

In the proposition below we are using the same notations that were introduced above Eq.~(7) of the main text.

\begin{proposition}
 Let $\ell\in\mathbb{N}$ and consider a bipartite cut $\mathcal{A}_m\otimes\mathcal{B}_{n-m}$. Then, the function
$$
 f_{\ell}^{\mathcal{A}_m |\mathcal{B}_{n-m}}(\psi)\equiv{\rm Tr}\left[\left(UAVA^T\right)^{\ell}\right]
$$
 is an SL-invariant polynomial of degree $2\ell$.
\end{proposition}

\begin{proof}
Let $g\in G$. We need to show that $f(g\psi)=f(\psi)$, where in this proof we removed from $f$ the subscript and superscript for the simplicity of the exposition. We can decompose $g=g_a\otimes g_b$ such
that $g_a\in SL(2,\mathbb{C})^{m}$ is acting on the space $\mathcal{A}_m$ and $g_b\in SL(2,\mathbb{C})^{n-m}$
is acting on $\mathcal{B}_{n-m}$. Further, denote by $G_{a}$ and $G_{b}$ the $m\times m$ and $(n-m)\times(n-m)$ 
square matrices whose elements are given by $\langle u_{j'}|g_a|u_j\rangle$ and $\langle v_{k'}|g_b|v_k\rangle$, respectively.
With this notations, replacing $|\psi\rangle$ with $g|\psi\rangle$, is equivalent to replacing $A$ with $G_{a}AG_{b}^{T}$.
We therefore have
\begin{align*}
f(g\psi)&={\rm Tr}\left[\left(UG_aAG_{b}^{T}VG_{b}A^TG_{a}^{T}\right)^{\ell}\right]\\
&={\rm Tr}\left[\left(G_{a}^{T}UG_aAG_{b}^{T}VG_{b}A^T\right)^{\ell}\right]
\end{align*}
where we have used the invariance of the trace under cyclic permutation. To complete the proof, we show now that
$$
G_{a}^{T}UG_a=U\;\;\text{and}\;\;G_{b}^{T}VG_b=V\;.
$$
Indeed,
$$
U_{jj'}=(u_j,u_{j'})_m=(g_au_j,g_au_{j'})_m=\left(G_{a}^{T}UG_a\right)_{jj'}\;\;,
$$
where we have used the invariance of the bilinear form $(\cdot,\cdot)_m$ under the action of $g_a$.
In the same way, one can prove that $G_{b}^{T}VG_b=V$. This completes the proof.
\end{proof}

As discussed in the main text, the above simple SLIPs are not linearly independent.
For odd number of qubits, there are no homogenous SLIPs with degree 2, and the lowest degree in this case is 4 (e.g.~\cite{W}).
In this case the number of homogenous SLIPs of degree 4 is $(2^{n-1}-1)/3$ (see below), whereas 
the number of our polynomials $f_{\ell=2}^{\mathcal{A}_q |\mathcal{B}_{n-q}}$ of degree 4 is equal to 
the number of inequivalent bipartite cuts $\mathcal{A}_q\otimes\mathcal{B}_{n-q}$ which is given by
$$
\sum_{q=1}^{\frac{n-1}{2}}{n\choose q}\;=2^{n-1}-1\;,
$$
where we assumed that $n$ is odd. Note that for any bipartite cut $\mathcal{A}_q\otimes\mathcal{B}_{n-q}$
there exists a bipartite cut of the form $\mathcal{A}_{n-q}\otimes\mathcal{B}_{q}$ which leads to the exact same polynomials. 
Thus, $q$ in the sum does not exceed $(n-1)/2$. Hence, our $2^{n-1}-1$ polynomials of degree 4, $f_{\ell=2}^{\mathcal{A}_q |\mathcal{B}_{n-q}}$, must be linearly dependent.
Nonetheless, the over-complete set of $2^{n-1}-1$ polynomials  $\{f_{\ell=2}^{\mathcal{A}_q |\mathcal{B}_{n-q}}\}$ span the set of all SLIPs of degree 4. 

\section{(iv) Schur-Weyl Duality}

The purpose of this short subsection is to give a very brief introduction to the Schur-Weyl duality, 
in the context of the work presented in this paper. For a more detailed analysis we refer the reader to 
a more standard book on representation theory of finite groups (c.f.~\cite{GW}). 

We consider the group $SL(m,\mathbb{C})$ acting on $\bigotimes^{k}\mathbb{C}^{m}$ by
the usual tensor action. The representations of $SL(m,\mathbb{C})$ appearing are determined
by their highest weight which is an $m$-tuple $\lambda=(\lambda_{1}%
,...,\lambda_{m})\in\mathbb{Z}^{m}$ with $\lambda_{1}\geq...\geq\lambda_{m}$.
We denote by $F^{\lambda}$ a fixed representation with highest weight
$\lambda$. Then $F^{\lambda}$ is equovalent with $F^{\mu}$ if and only if
$\lambda-\mu=(\delta,...,\delta)$ with $\delta\in\mathbb{Z}$. Also
$F^{\lambda}$ is a constituent of $\bigotimes^{k}\mathbb{C}^{m}$ if an only if
it is equivalent to $F^{\mu}$ with $\mu_{1}+...+\mu_{m}=k$ and $\mu_{m}\geq0$.
We will say that $\lambda=(\lambda_{1},...,\lambda_{r})\in\mathbb{Z}^{r}$ with
$\lambda_{1}\geq...\geq\lambda_{r}$ is dominant. Also if $\lambda$ is dominant
then we will call it a partition of $k$ if $\lambda_{1}+...+\lambda_{r}=k$ and
$\lambda_{r}\geq0$. Obviously, we can make the parametrization unambiguous by
using $\lambda$ with $\lambda_{m}=0$. But we note that if $\lambda,\mu
\in\mathbb{Z}^{m}$ are partitions of $k$ and if $F^{\lambda}$ and $F^{\mu}$
are equivalent then $\lambda=\mu$. In the literature, $\mu$ denotes a partition
of $k$ and if $l$ denotes the last index for which $\mu_{l}>0$ then the indices
$l+1,...,m$ are dropped. These partitions parametrize the equivalence classes
represenations of the symmetric group $S_{k}$. We fix a fixed realization of
the class that is parametrized by $\lambda$ and denote it by $V^{\lambda}$. If
we need to keep track of $k$ and $n$ we will write $V_{k}^{\lambda}$ and
$F_{n}^{\lambda}$.

We note that the group $S_{k}$ acts on $\bigotimes^{k}\mathbb{C}^{m}$ by
permuting the indices. This action commutes with that of $SL(m,\mathbb{C})$. Schur-Weyl
duality says that as a representation of $S_{k}\times SL(m,\mathbb{C})$ the space
$\bigotimes^{k}\mathbb{C}^{m}$ decomposes into
\[
\bigoplus_{\lambda}V_{k}^{\lambda}\otimes F_{m}^{\lambda}
\]
with the sum over all partitions $\lambda=(\lambda_{1},...,\lambda_{m})$ of
$k$. We note that if $m\geq k$ that this is a sum over all partitions of $k$
(after we have dropped all of the entries that are $0$).

This decomposition allows us to describe the projection of $\bigotimes
^{k}\mathbb{C}^{m}$ into $V_{k}^{\lambda}\otimes F_{m}^{\lambda}$ using either
a projection gotten by integrating the character of $F^{\lambda}$ restricted
to $SU(m)$ or summing the character of $V^{\lambda}$ over $S_{k}$. In other
words ($d\mu$ is the unique invariant probability measure on $SU(m)$)
\begin{align*}
P^{\lambda}&=\dim F^{\lambda}\int_{SU(n)}\overline{\chi_{F^{\lambda}}(u)}%
\otimes^{k}ud\mu(u)\\
&=\frac{\dim V^{\lambda}}{k!}\sum_{\sigma\in S_{k}}\chi_{V^{\lambda}}%
(\sigma)\sigma.
\end{align*}
An example of this is if $k=mr$ with $r$ a positive integer and $\lambda
=r(1,...,1)$ ($m$ ones) then $P^{\lambda}$ is the projectiion onto the 
$SL(m,\mathbb{C})$-invariants in $\bigotimes^{k}\mathbb{C}^{m}$.

\section{(v) Examples: using Eqs.(9,10) of the main text to construct SLIPs}

Formula (10) (of the main text) for $k=2$ necessitates $r=1$ and $m=2$. Thus, in this case $P_{m,k}$ is acting on 
$\mathbb{C}^{2}\otimes\mathbb{C}^{2}$ and is given by its action on basis elements as:
\[
v_{1}\otimes v_{2}\longmapsto\frac{1}{2}(v_{1}\otimes v_{2}-v_{2}\otimes
v_{1}).
\]
This implies that Formula (9) says for the case of $2s$ qubits
\[
(x_{1}\otimes x_{2}\otimes\cdots\otimes x_{2s})\otimes(y_{1}\otimes
y_{2}\otimes\cdots\otimes y_{2s})\longmapsto(x_{1}\otimes y_{1})\otimes
(x_{2}\otimes y_{2})\otimes\cdots\otimes(x_{2s}\otimes y_{2s})\longmapsto
\]%
\[
\frac{1}{2^{2s}}(x_{1}\otimes y_{1}-y_{1}\otimes x_{1})\otimes\cdots
\otimes(x_{2s}\otimes y_{2s}-y_{2s}\otimes x_{2s})\longmapsto\frac{1}{2^{2s}%
}\sum_{S\subset\{1,2....,2s\}}(-1)^{|S|}u_{S}\otimes u_{S^{c}}.
\]
Here the first arrow is the transpose as is the last. Also $S^{c}%
=\{1,2,...,2s\}-S$, $|S|$ is the size of $S$ and $u_{S}=z_{1}\otimes
z_{2}\otimes\cdots\otimes z_{2s}$ with $z_{j}=y_{j}$ if $j\in S$ and
$z_{j}=x_{j}$ if $j\notin S$. This yields the well known formula for the
degree $2$ invariant in any even number of qubits.

Formula (10) for $k=4$ plays a role only for $r=1$ and $m=4$ or $r=2$ and
$m=2$. In the first case Formula (10) implies that $P_{m,k}$ is acting on 
$\mathbb{C}^{4}\otimes\mathbb{C}^{4}\otimes\mathbb{C}^{4}\otimes\mathbb{C}^{4}$ and is given by its action on basis elements as:
\[
v_{1}\otimes v_{2}\otimes v_{3}\otimes v_{4}\longmapsto\frac{1}{24}%
\sum_{\sigma\in S_{4}}sgn(\sigma)v_{\sigma1}\otimes v_{\sigma2}\otimes
v_{\sigma3}\otimes v_{\sigma4}.
\]
We leave it to the reader to write out Formula (9) to get the degree $4$ invariant
for an arbitrary number of $4$-dimensional qudits. 

In the case of $r=2$ and $m=2$ then the
pertinent representation, $\lambda$, of $S_{4}$ corresponds to the partition
$2,2.$ The dimension of $\lambda$ is $2$. The conjugacy classes in $S_{4}$ are
given by $C_{4}$ the cycles of length $4$ (e.g. $(1234)$), $C_{3,1}$ the
cycles of length $3$, $C_{2,2}$ the products of disjoint transpositions (e.g.
$(12)(34)$), $C_{2,1,1}$ the transpositions,$C_{1,1,1,1}$ the identity
element. We have
\[
\chi_{\lambda}(\sigma)=\left\{
\begin{array}
[c]{c}%
0\text{ if }\sigma\in C_{4}\\
-1\text{ if }\sigma\in C_{3,1}\\
2\text{ if }\sigma\in C_{2,2}\\
0\text{ if }\sigma\in C_{2,1,1}\\
2\text{ if }\sigma=I
\end{array}
\right.  .
\]
We note that $|C_{3,1}|=8$,$|C_{2,2}|=3$.   Formula (10) implies that $P_{m,k}$ is acting on 
$\mathbb{C}^{2}\otimes\mathbb{C}^{2}\otimes\mathbb{C}^{2}\otimes\mathbb{C}^{2}$ and is given by its action on basis elements as:
\[
v_{1}\otimes v_{2}\otimes v_{3}\otimes v_{4}\longmapsto\frac{1}{6}v_{1}\otimes
v_{2}\otimes v_{3}\otimes v_{4}-\frac{1}{12}\sum_{\sigma\in C_{3,1}}%
v_{\sigma1}\otimes v_{\sigma2}\otimes v_{\sigma3}\otimes v_{\sigma4}+\frac
{1}{6}\sum_{\sigma\in C_{2,2}}v_{\sigma1}\otimes v_{\sigma2}\otimes
v_{\sigma3}\otimes v_{\sigma4}.
\]
We leave it to the reader to write out Formula (9) for $n$ qubits.

\section{(vi) The dimension formula of the space of SLIPs of fixed degree}

The reference that we will use for this section is Chapter 9 in~\cite{GW}. In
this section we will consider the case of $n$--qudits all of the same
dimension, say, $m.$ Thus we are looking at $\mathcal{H}_{n}=\otimes
^{n}\mathbb{C}^{m}$ and $G=SL(m,\mathbb{C})\times\cdots\times SL(m,\mathbb{C}%
)$, $n$ copies. we analyse the space $S^{k}(\mathcal{H}_{n})^{G}$ as the
$S_{k} $--invariants in $\left(  \otimes^{k}\left(  \otimes^{n}\mathbb{C}%
^{m}\right)  \right)  ^{G}$. Using the map $\left\vert u\right\rangle
\longmapsto$ $\left\vert u\right\rangle ^{T}$ defined in (8) we are
considering%
\[
(\otimes^{n} (\otimes^{k}\mathbb{C}^{m})^{SL(m,\mathbb{C})}) ^{S_{k}%
}\overset{}{(\ast)}
\]
with $S_{k}$ acting by the $n$--fold tensor product of the action on
$\otimes^{k}\mathbb{C}^{m}$ which is the action of permuting the tensor
factors. $(\ast)$ is $0$ unless $k=rm$ for some $r\in\mathbb{Z}_{\geq0}$. As
we observed in subsection (iv) the $SL(m,\mathbb{C})$--invariants in
$\otimes^{k}\mathbb{C}^{m}$ define the representation $V^{\lambda}$ of $S_{k}$
with $\lambda=(r,r,...,r)\in\left(  \mathbb{Z}_{\geq0}\right)  ^{m}$. \ Thus
as far as $S_{k}$ is concerned
\[
\left(  \otimes^{k}\left(  \otimes^{n}\mathbb{C}^{m}\right)  \right)
^{G}=\otimes^{n}V^{\lambda}.
\]
Thus if $k=rm$ and $\lambda$ are as above then%
\[
\dim S^{k}(\mathcal{H}_{n})^{G}=\left(  \otimes^{n}V^{\lambda}\right)
^{S_{k}}.
\]
This says that if $\chi_{\lambda}$ is the character of $V^{\lambda}$ then%
\[
\dim S^{k}(\mathcal{H}_{n})^{G}=\frac{1}{k!}\sum_{\sigma\in S_{k}}%
\chi_{\lambda}(\sigma)^{n}\text{.}
\]
We recall that characters are constant on conjugacy classes (the sets of the
form $\{\mu\sigma\mu^{-1}|\mu\in S_{k}\}$). The conjugacy classes of $S_{k}$
can be described by partitions using the cycle decomposition of an
permutation. The (standard) algorithm for this is start with a permutation
$\sigma$. Consider $\sigma^{j}1$. Then there must be a first $k_{1}$ such that
$\sigma^{k_{1}-1}1\neq1$ but $\sigma^{k_{1}}1=1$. Set $m_{1,j}=\sigma^{j}1$
for $j=0,...,k_{1}-1$ then on the set $K_{1}=\{m_{1,0},...,m_{1,k_{1}-1}\}$
the permutation $\sigma$ acts by the $k_{1}$--cycle $(m_{1,0}m_{1,1}\dots
m_{1k_{1}-1})$. If $k_{1}=k$ then we are done. Otherwise choose the smallest
element of $\{1,...,k\}$ that is not in $K_{1}$ and denote it $m_{2,0}$ then
$\sigma^{j}m_{2,0}=m_{2,0}$ for the first time for $j=k_{2}$. Set
$m_{2,j}=\sigma^{j}m_{2,0}$. We now have a second cycle whose entries are
disjoint from $K_{1}$,$(m_{2,0}m_{2,1}\dots m_{2k_{2}-1})$ action on the set
$K_{2}=\{m_{2,j}|j=0,...,k_{2}-1\}$. If $k_{1}+k_{2}=k$ then we are done
otherwise do the same procedure on $\{1,...,n\}-\left(  K_{1}\cup
K_{2}\right)  $, etc. Finally put the $k_{1},...,k_{l}$ in decreasing order.
This yields a partition of $k$ corresponding to each $\sigma\in S_{k}$. It is
not hard to see that $\sigma$ and $\mu$ are conjugate if and only if they have
the same cycle decomposition.

If $\mu$ is a partition of $k$ then let $c_{\mu}$ denote the corresponding
conjugacy class and let $Ch_{\mu}$ be the characterisitic function of the
subset $c_{\mu}$ of $S_{k}$. That is,
\[
Ch_{\mu}(\tau)=\left\{
\begin{array}
[c]{c}%
1\text{ if }\tau\in c_{\mu}\\
0\text{ otherwise}%
\end{array}
\right.  .
\]
\ Then the functions $Ch_{\mu}$ form a basis of the space of central functions
on $S_{k}$. Thus if $\lambda$ is a partitian of $k$ then%
\[
\chi_{\lambda}=\sum_{\mu}a_{\lambda\mu}Ch_{\mu}%
\]
where the sum is over the partitions of $k$. The numbers $\alpha_{\lambda,\mu
}$ arranged into a table give the character table of $S_{k}$. However the
conjugacy class partitions are usually given in the way we will describe them
next. Let $C_{\mu}$ be a conjugacy clasee then we reorganize the partition in
\emph{increasing} order and lay it out as $1^{m_{1}}2^{m_{2}}\cdots k^{m_{k}}$
where $m_{1}$ of the cycles are of length $1$ (fixed points), $m_{2}$ are of
length $2$,..., $m_{s}$ are of length $s$. Thus the permutation
\[
1\rightarrow2,2\rightarrow4,3\rightarrow1,4\rightarrow3,5\rightarrow
6,6\rightarrow5,7\rightarrow8,8\rightarrow7
\]
has $\mu=(4,2,2)$ which also corresponds to $1^{0}2^{2}4^{1}$ we will use the
two notations for partitions interchangably. The following is standard and
easily proved

\begin{lemma}
If $\mu$ is given as $1^{m_{1}}2^{m_{2}}\cdots k^{m_{k}}$ then
\[
\left\vert c_{\mu}\right\vert =\frac{k!}{\prod_{j=1}^{k}\left(  m_{j}%
!j^{m_{j}}\right)  }.
\]

\end{lemma}

Here is the result that we have been aiming at

\begin{proposition}
If $k\neq mr$ then $\dim S^{k}(\mathcal{H}_{n})^{G}=0$ otherwise set
$\lambda=(r,...,r)$ with $m$ entries. Then%
\[
\dim S^{k}(\mathcal{H}_{n})^{G}=\frac{1}{k!}\sum_{\mu}(a_{\lambda\mu}%
)^{n}|c_{\mu}|
=\sum_{\mu}\frac{(a_{\lambda\mu})^{n}}{\prod_{j=1}^{k}\left(  m_{j}!j^{m_{j}%
}\right)  }.
\]

\end{proposition}

The asymptotic formula follows from this. We will give it in full generality
and at the end of this subsection we will explain how the formulas look for qubits.

Let $(\pi,V)$ be an irreducible representation of a finite group, $\Gamma$. If
$\gamma\in\Gamma$ then the Schwarz inequality implies that%
\[
\left\vert \text{tr}(\pi(\gamma))\right\vert \leq\dim V
\]
with equality if and only if $\pi(\gamma)=\zeta I$ for some norm-one scalar
$\zeta$. The subgroup of all elements of $\Gamma$ such that $\pi(\gamma)$ is a
multiple of the identity is a normal subgroup. If $k>4$ then the only
non-trivial normal subgroup of $S_{k}$ is the alternating group. This implies
that if $k>4$ and if $\lambda$ is a partition of $k$ such that the
representation $V^{\lambda}$ is not a one dimensional representation then%
\[
\left\vert \chi_{\lambda}(\sigma)\right\vert <\chi_{\lambda}(I)=\dim
V^{\lambda}%
\]
if $\sigma$ is not the identity in $S_{k}$. If $k=4$ and $\lambda=(2,2)$ then
there is a Klein $4$--group of elements that map to the identity.

We have

\begin{corollary}
If $k>4$ then we have the following asymptotic formula when $k=rm$ and
$\lambda=(r,...,r)$%
\[
\lim_{n\rightarrow\infty}\frac{\dim S^{k}(\mathcal{H}_{n})^{G}}{(\dim
V^{\lambda})^{n}}=\frac{1}{k!}.
\]
If $k=4$ and $\lambda=(2,2)$ (so $m=2)$ then then the formula is%
\[
\lim_{n\rightarrow\infty}\frac{\dim S^{4}(\mathcal{H}_{n})^{G}}{2^{n}}%
=\frac{1}{6}.
\]

\end{corollary}

\begin{proof}
If $k>4$ then we note that if $\beta=(1,1,...,1)$ is the partition
correspnding to the identity in $S_{k}$ then $a_{\lambda,\beta}=\dim
V^{\lambda}>\left\vert a_{\lambda,\mu}\right\vert $ for any other partition,
$\beta$. \ Thus%
\[
\frac{\dim S^{k}(\mathcal{H}_{n})^{G}}{\left(  \dim V^{\lambda}\right)  ^{n}%
}=\frac{1}{k!}\sum_{\mu}(\frac{a_{\lambda\mu}}{a_{\lambda,\beta}})^{n}|c_{\mu
}|.
\]
Since $\lim_{n\rightarrow\infty}\frac{a_{\lambda\mu}}{a_{\lambda,\beta}%
}=\delta_{\beta,\mu}$ and $\left\vert c_{\beta}\right\vert =1$, this proves
the result if $k>4$. If $k=4$ and $\lambda=(2,2)$ then $\dim V^{\lambda}=2$
and there are $4$ elements that yield the value $\pm I$.
\end{proof}

We note that (c.f.~\cite{GW} Corollary 9.1.5) if $\lambda=rm$ then
\[
\dim V^{\lambda}=\frac{k!\prod_{1\leq i<j\leq m}(j-i)}{\prod_{j=1}%
^{m}(r+m-j)!}.
\]
Thus if $m=2$ and $k=2r$ then $\dim V^{\lambda}=\frac{1}{r+1}\binom{2r}{r}$
the $r$--th Catalan number. For $k=rm$ and $\lambda=(r,...,r)$ ($m$--factors)
then the dimension should be thought of as a mutinomial generalization of the
Catalan number.

\section{(vii) An alternative way to express all SLIPs in the case of $n$ qubits}

Our expression of $P_{m,k}$ in Eq.~(10) of the main text involves in general many terms and therefore is somewhat cumbersome.
Hence, in this section we show an alternative way to express \emph{all} SLIPs of degree $k$ in the 
important case of $n$-qubits. This new technique has the advantage that it is computationally efficient for small degrees.
We have already seen that for the qubit case we can express \emph{some} of the SLIPs
very elegantly with Eq.(10). However, that construction is only partial as it does not consists of all SLIPs. 
Hence, as an application to our new expression in proposition~\ref{sl2} (see below), we give at the end of this section an example of the unique homogeneous SLIP of degree 6 in 5 qubits; this SLIP can not be written in the form of Eq.(7) of the main text. 

\subsubsection{An SL(2,$\mathbb{C}$) technique}

In this section $\mathcal{H}_{n}$ denotes $\otimes^{n}%
\mathbb{C}
^{2}$ and $G_{n}=SL(2,%
\mathbb{C}
)\otimes\cdots\otimes SL(2,%
\mathbb{C}
)$ $n$--copies. We set%
\[
d_{n,k}=\dim P^{k}(\mathcal{H}_{n})^{G_{n}}\text{.}%
\]
We consider $\mathcal{H}_{n+1}=\left\vert 0\right\rangle \otimes
\mathcal{H}_{n}\oplus\left\vert 1\right\rangle \otimes\mathcal{H}_{n}$. Then
the character of the invariants for
\[
P^{k}(\mathcal{H}_{n})^{I\otimes G_{n}}%
\]
as an $SL(2,%
\mathbb{C}
)\otimes I_{\mathcal{H}_{n}}$ representation restricted to
\[
T=\left\{  \left[
\begin{array}
[c]{cc}%
q & 0\\
0 & q^{-1}%
\end{array}
\right]  |q\in\mathbb{C}-\{0\}\right\}
\]
is
\[
q^{k}d_{n,k}+q^{k-2}\dim(P^{k-1}(\mathcal{H}_{n})\otimes P^{1}(\mathcal{H}%
_{n}))^{G_{n}}+q^{k-4}\dim(P^{k-2}(\mathcal{H}_{n})\otimes P^{2}%
(\mathcal{H}_{n}))^{G_{n}}+...+q^{-k}d_{n,k}.
\]
Here the action of $G_{n}$ is the tensor product action on the factors. We
note that the coefficient of $q^{j}$ is the same as $q^{-j}$ in this
expression. We also note that $d_{n,k}=0$ if $k$ is odd. Thus we may assume that
$k$ is even. We observe that the character of the $\frac{k}{2}$--spin
representation of $SL(2,%
\mathbb{C}
)$ restricted to $T$ is
\[
q^{k}+q^{k-2}+...+q^{0}+...+q^{-k}.
\]
This implies that if $k=2r$ then%
\[
d_{n+1,k}=\dim\left(  \left(  P^{r}(\mathcal{H}_{n})\otimes P^{r}%
(\mathcal{H}_{n})\right)  ^{G_{n}}\right)  -\dim\left(  \left(  P^{r+1}%
(\mathcal{H}_{n})\otimes P^{r-1}(\mathcal{H}_{n})\right)  ^{G_{n}}\right)  .
\]
it also implies that if $f$ $\neq0$ is a homogeneous polynomial of degree $2r$
in $n+1$ qubits that is invariant under $G_{n+1}$ then if we think of $f$ as a
polynomial in two copies of $n$ qubits it must have a non-zero component in
$\left(  P^{r}(\mathcal{H}_{n})\otimes P^{r}(\mathcal{H}_{n})\right)  ^{G_{n}%
}$. For simplicity we will concentrate on the case when $n$ is even (so $n+1$
is odd). This hypothesis implies that $\mathcal{H}_{n}$ there exists a (unique
up up to scalar) $G_{n}$--invariant symmetric complex bilinear form $($ \ ,
\ \ $)$ on $\mathcal{H}_{n}$. Let $v_{i}$ with $i=0,1,...,N=2^{n}-1$ be a an
orthonormal basis with respect to this form. We think of the first copy,
$\left\vert 0\right\rangle \otimes\mathcal{H}_{n}$, as having elements $\sum
x_{j}v_{j}$ and the second, $\left\vert 1\right\rangle \otimes\mathcal{H}_{n}%
$, $\sum y_{j}v_{j}$. Then the action of the Lie algebra first
$SL(2,\mathbb{C})$--factor in $G_{n+1}$ (that is $g\otimes I_{\mathcal{H}_{n}%
}$) \ is given by%
\[
X=\sum_{j=0}^{N}x_{j}\frac{\partial}{\partial y_{j}},\;\;\;
Y=\sum_{j=0}^{N}%
y_{j}\frac{\partial}{\partial x_{j}},\;\;\;
H=\sum_{j=0}^{N}x_{j}\frac{\partial
}{\partial x_{j}}-\sum_{j=0}^{N}y_{j}\frac{\partial}{\partial y_{j}}.
\]
With commutation relations%
\[
\lbrack X,Y]=H,\;\;[H,X]=2X,\;\;[H,Y]=-2Y.
\]
We note that $X\longleftrightarrow\left[
\begin{array}
[c]{cc}%
0 & 1\\
0 & 0
\end{array}
\right]  ,Y\longleftrightarrow\left[
\begin{array}
[c]{cc}%
0 & 0\\
1 & 0
\end{array}
\right]  ,H\longleftrightarrow\left[
\begin{array}
[c]{cc}%
1 & 0\\
0 & -1
\end{array}
\right]  $ so the Casimir operator is given by
\[
C=XY+YX+\frac{1}{2}H^{2}.
\]
Since $XY=YX+H$ we have%
\[
C=2YX+\frac{1}{2}(H+1)^{2}-\frac{1}{2}.
\]
We replace $C$ with $L=YX+\frac{1}{4}(H+1)^{2}$ and note that this operator
commutes with the action of $G_{n+1}$.

We note that if we use the same formulas for $X$,$Y$ and $H$ on the
polynomials in $x_{0},x_{1},...,x_{N}.y_{0},...,y_{N}$ with $N$ an arbitrary
positive integer then we have the same commutation relations. Let
$P_{p,q,N+1}$ be the space of all polynomials that are homogeneous of degree
$p$ in $x_{0},...,x_{N}$ and degree $q$ in $y_{0},...,y.$ Let $X_{N+1}%
,Y_{N=1},H_{N+1}$ be the corresponding operators (as above). The following is obvious.

\begin{lemma}
Let $0\leq k<N$ and let $T_{k,N}:P_{p,q,N+1}\rightarrow P_{p,q,k+1}$ be given
by
\[
T_{k+1,N+1}(f)(x_{0},...,x_{k},y_{0},...,y_{k})=f(x_{0},...,x_{k}%
,0,...0,y_{0},...,y_{k}.0,...,0)
\]
with $N-k$ zeros. Then if $Z=X,Y$ or $H$ then%
\[
T_{k,N}(Z_{N+1}f)=Z_{k+1}T_{k,N}(f)\text{.}%
\]

\end{lemma}

Using the classification of the irreducible finite dimensional representations
of $SL(2,\mathbb{C})$ we see that on the spin $\frac{m}{2}$ (i.e. $m+1$
dimensional representation), $F^{m}$, the operator $L$ acts by $\frac
{(m+1)^{2}}{4}I$. Thus in a representation, $W$, for $SL(2,\mathbb{C)}$ that
is a sum with multiplicity of $F^{0},F^{2},...,F^{2r}$ then the projection of
$W$ onto the invariants is given by%
\[
\frac{(-1)^{r}}{r!(r+1)!}%
{\displaystyle\prod\limits_{j=1}^{r}}
(YX+\frac{1}{4}(H+1)^{2}-\frac{(2j+1)^{2}}{4}).
\]
Using this, taking $W$ as above we get the main result of this section:
\begin{proposition}\label{sl2}
If $W_{0}=\{w\in W|Hw=0\}$ then the space of $SL(2,\mathbb{C})$ invariants is
\[
{\displaystyle\prod\limits_{j=1}^{k}}
(YX-j(j+1))W_{0}.
\]
\end{proposition}

\subsubsection{Example: Unique SLIP of degree 6 in 5 qubits} 

Example the invariant in $5$ qubits of degree $6$. In $4$ qubits we take the
basis $u_{r}\otimes u_{s}$ with
\begin{align*}
&u_{0}=\frac{1}{\sqrt{2}}\left(  \left\vert 00\right\rangle +\left\vert
11\right\rangle \right)\;\;,\;\;
u_{1}=\frac{i}{\sqrt{2}}\left(  \left\vert 00\right\rangle -\left\vert
11\right\rangle \right)\;,\\
& u_{2}=\frac{i}{\sqrt{2}}\left(  \left\vert 01\right\rangle +\left\vert
10\right\rangle \right)\;\;,\;\;
u_{3}=\frac{1}{\sqrt{2}}\left(  \left\vert 00\right\rangle -\left\vert
11\right\rangle \right)  .
\end{align*}
That is, the Bell basis in $2$ qubits. These bases are orthonormal with
respect to both the Hilbert space inner product and ( suitably normalized) the
respectively $G_{2}$ or $G_{4}$-invariant complex bilinear, symmetric form. We
order the basis so that $v_{j}=$ $u_{j}\otimes u_{j}$, $j=0,1,2,3$.

Let $z=\sum_{r,s}z_{r,s}u_{r}u_{s}$,$Z=[z_{r,s}],$ $f(z)=\det[Z]$ and
$g(z)=tr((ZZ^{T})^{2})$. We order the basis so that $v_{j}=$ $u_{j}\otimes
u_{j}$, $j=0,1,2,3$ and consider $f$ to be $f(x_{0},...,x_{15})$ and $g$ to be
$g(y_{0},...,y_{15})$. We set
\[
w=\sum_{j}\frac{\partial f(x)}{\partial x_{j}}\frac{\partial g(y)}{\partial
y_{j}}%
\]
then this is an invariant polynomial for $G_{4}$ acting on two copies of $4$
qubits of bidegree $3,3$. We note that
\[
T_{4,16}(w)=4\left(  x_{0}x_{1}x_{2}y_{3}^{3}+x_{0}x_{1}y_{2}^{3}x_{3}%
+x_{0}y_{1}^{3}x_{2}x_{3}+y_{0}^{3}x_{1}x_{2}x_{3}\right)  .
\]
Set $$\phi(x,y)=x_{0}x_{1}x_{2}y_{3}^{3}+x_{0}x_{1}y_{2}^{3}x_{3}+x_{0}y_{1}%
^{3}x_{2}x_{3}+y_{0}^{3}x_{1}x_{2}x_{3}\;.$$ Then
$$
(Y_{4}X_{4}-12)(Y_{4}X_{4}-6)(Y_{4}X_{4}-2)\phi(x,y)
=36(-\phi(x,y)+\phi
(y,x)+\mu(x,y)-\mu(y,x))
$$
with
\begin{align*}
\mu(x,y)&=x_{0}y_{0}^{2}(y_{1}x_{2}x_{3}+x_{1}y_{2}x_{3}+y_{1}x_{2}x_{3})
+x_{1}y_{1}^{2}(y_{0}x_{2}x_{3}+x_{0}y_{2}x_{3}+x_{0}x_{2}y_{3})\\
&+x_{2}y_{2}^{2}(y_{0}x_{1}x_{3}+x_{0}y_{1}x_{3}+x_{0}x_{1}y_{3})
+x_{3}y_{3}^{2}(y_{0}x_{1}x_{2}+x_{0}y_{1}x_{2}+x_{0}x_{1}y_{2}).
\end{align*}
We have
\begin{corollary}
The following is up to scalar the unique degree $6$ invariant for $5$ qubits.%
\[
(Y_{16}X_{16}-12)(Y_{16}X_{16}-6)(Y_{16}X_{16}-2)w(x,y).
\]
\end{corollary}

\begin{proof}
We know that the space of such invariants is dimension 1 and since the
retriction of the element in the statement is non-zero by the above it must be
the invaraint.
\end{proof}

\section{(viii) Mathematica code for the case of qubits}

The purpose of this section is to give a listing of Mathematica code to
calculate $d_{n,k}$ as a function of $n$ if $k$ is fixed (here we look only at
qubits). In the body of the letter we give two examples. Using the code below
and taking $k=10$ we have %
\[
272160+28448(-3)^{n}+766080(-1)^{n}+338751(2)^{n}+14175(-1)^{n}2^{2+n}+
\]
\[
11200(3)^{n+1}+35(2)^{n+1}3^{n+2}+315(-1)^{n}4^{n+3}+189(-2)^{n}%
5^{n+1}++45(14)^{n}+42^{n}%
\]
over
\[
10!=3628800.
\]
The calculation was almost instantaneous. In our code below we will output the
expression in powers of $q$. We use the notation in subsection v of this
appendix. We first explain the method used in the code. If $k=2r$ then the
code implements (the sum below is on the partitions of $k$)%
\[
\frac{1}{k!}\sum_{\mu}(a_{(r,r),\mu})^{q}|c_{\mu}|.
\]
One can read this off of tables for small values of $k$ or use a standard
mathematical package. However, the included code is designed to efficiently do
this specific calculation. It should be easily converted to a
lower level language (such as C or C++). However, the output becomes immense
for $k$ larger than say $30$ (a calculation that takes about 1 minute on a 4
year old PC). The code below has two main functions.

Mult[k,n] which calculates $d_{n,k}$ for specific values of $k$ and $n$.

Multq[k] which calculates $d_{q,k}$ as a function of $q$ with $k$ fixed.

The algorithm for calculating $a_{(r,r),\mu}$ uses Theorem 9.1.4 in~\cite{GW} with
the \textquotedblleft$n$\textquotedblright\ in that theorem equal to $2.$And
the fixed point calculation in the formula given by the function FP[p,q,x]
\ (here $x$ corresponds to a conjugacy class). The number  $|c_{\mu}|$ is
calculated in the function CST[x]. The main weakness in the code is the
generation of the partitions Prt[k]. Since the number of partitions of is
$O(e^{C\sqrt{k}})$ this will be a bottleneck no matter what one does to
streamline the code.

The Mathematica code:
\begin{lstlisting}
P[n_, k_] :=
     Module[{L, M = {}, S, i, j}, If[k > n, Return[P[n, n]]]; If[k == 0, 
    Return[{{}}]];
    If[k == 1, Return[{Table[1, {i, 1, n}]}]];
    For[i = 1, i <= k, i++, L = P[n - i, i]; S = {};
      For[
      j = 1, j <= Length[L], j++, S =
         Append[S, Prepend[L[[j]], i]]]; M = Union[M, S]]; M]

(*Partions of n in Lex order *)		 
PRT[n_] := P[n, n]
(* Calculates the order of the fixed point set of x on S_{p+q}/S_p S_q *)
FP[p_, q_, x_] := Module[{a = x[[1]], r = p, s = 
    q}, If[Length[x] == 0, Return[0]]; If[p < q, r = q; s = 
    p]; If[s == 0, Return[1]]; If[Max[x] == 1, Return[Binomial[r + s, r]]]; 
	If[a > r, Return[0]];
    If[a <= s, Return[FP[r - a, s, Delete[x, 1]] + FP[r, s - a, 
    Delete[x, 1]]], Return[FP[r - a, s, Delete[x, 1]]]]]
(* If x is a partition of n then n!/CSTD[x] is the size of the conjugacy class of x *)	
CSTD[x_] := Module[{y,
   j, i, k}, If[x == {}, Return[1]]; If[Length[x] == 1, Return[x[[1]]]];
    If[x[[1]] > x[[2]], Return[x[[1]]CSTD[Delete[x, 1]]]];
       If[Length[x] == 2, Return[2*x[[1]]^2]];
    k = 2; y = Delete[x, 1]; y = Delete[y, 1];
    For[j = 3, j <= Length[x], j++, If[y[[1]] == x[[1]], k++; y = 
    Delete[y, 1], Break[]]];
    If[k == Length[x], Return[x[[1]]^k k!], Return[x[[1]]^k k!CSTD[y]]]]
	
(* This calculates the dimension of the invariants of degree k in n qubits *)	

Mult[k_, n_] := Module[{L, M, i}, If[Mod[k, 2] == 1, Return[0]];
    L = PRT[k];
    Sum[(FP[k/2, k/2, L[[i]]] - FP[k/2 + 1, k/2 - 1, L[[
    i]]])^n/CSTD[L[[i]]], {i, 1, Length[L]}]]
	
(* This calculates the dimension of the invariants of degree k in q qubits as a function of q *)
Multq[k_] := Module[{L, M, i, r, f = 0}, 
	If[Mod[k, 2] == 1, Return[0]];
    L = PRT[k];
    For[i = 1, i <= Length[L], i++, 
	r = FP[k/2, k/2, L[[i]]] - FP[k/2 + 1, k/2 - 1, L[[i]]]; 
	If [r != 0, f = f + r^q/CSTD[L[[i]]]]]; 
	f]
\end{lstlisting}

\end{document}